\documentclass[%
 aip,
 jmp, 
 amsmath,amssymb,
 reprint, onecolumn 
]{revtex4-2}

\usepackage{graphicx}
\usepackage{dcolumn}
\usepackage{bm}
\usepackage[final]{changes}

\usepackage[utf8]{inputenc}
\usepackage[T1]{fontenc}
\usepackage{mathptmx} 
\usepackage{etoolbox} 
\usepackage[english]{babel}
\usepackage{enumerate}
\usepackage{mathtools}
\usepackage{csquotes}
\usepackage{dsfont}
\usepackage{xcolor}
\usepackage{bbm}
\usepackage{comment}
\usepackage{amsthm}
\usepackage{csquotes}

\makeatletter
\def\@email#1#2{%
 \endgroup
 \patchcmd{\titleblock@produce}
  {\frontmatter@RRAPformat}
  {\frontmatter@RRAPformat{\produce@RRAP{*#1\href{mailto:#2}{#2}}}\frontmatter@RRAPformat}
  {}{}
}%
\makeatother

\newcommand{\RN}[1]{\uppercase\expandafter{\romannumeral #1\relax}}
\newcommand*\diff{\mathop{}\!\mathrm{d}}

\newtheorem{Th}{Theorem}[section]
\newtheorem{Pro}[Th]{Proposition}
\newtheorem{Cor}[Th]{Corollary}
\newtheorem{Lem}[Th]{Lemma}

\newtheorem{Rem}[Th]{Remark}

\renewcommand{\epsilon}{\varepsilon}

\newcommand{\T}{{\mathcal T}}

\renewcommand\S{{\mathcal{S}}}

\renewcommand\L{{\mathcal L}}

\newcommand\I{{\mathcal I}}

\newcommand\E{{\mathcal E}}

\newcommand\NN{{\mathbb N}}
\newcommand\ZZ{{\mathbb Z}}

\newcommand\RR{{\mathbb R}}
\newcommand\R{{\mathcal R}}

\newcommand\C{{\mathcal C}}

\newcommand\ii{{\infty}}

\newcommand\D{{\mathcal D}}

\def\per{{\rm per }}
\def\loc{{\rm loc }}

\def\unif{ \text{\rm  unif }  }

\newcommand{\norm}[1]{\left\| #1\right\|}
\newcommand{\set}[1]{\left\{ #1\right\}}
\newcommand{\bra}[1]{\left( #1\right)}
\newcommand{\av}[1]{\left| #1\right|}
\newcommand{\abs}[1]{\left| #1\right|}

\renewcommand{\phi}{\varphi}

\def\bm{\mathbf{m}}

\def\div{{\rm div \;}}
\def\NN{{\mathbb N}}
\def\ZZ{{\mathbb Z}}
\def\RR{{\mathbb R}}

\renewcommand{\leq}{\leqslant}

\renewcommand{\geq}{\geqslant}

\usepackage{todonotes}
\begin{document}

\title{Homogenization of 2D materials in the Thomas--Fermi--von Weizs{\"a}cker theory}

\author{Saad Benjelloun}
 \email{saad.benjelloun@devinci.fr}
 \affiliation{De Vinci Higher Education, De Vinci Research Center, Paris, France}
 \affiliation{Makhbar Mathematical Sciences Research Institute, Casablanca, Morocco}

\author{Salma Lahbabi}
 \email{s.lahbabi@ensem.ac.ma}
 \affiliation{Equipe de Mathématiques Appliquées, LARILE, ENSEM, Hassan II University of Casablanca, Morocco}
\affiliation{College of Computing, Université Mohamed 6 Polytechnique, Benguerir, Morocco}
 
\author{Abdelqoddous Moussa}
 \email{abdelqoddous.moussa@um6p.ma}
 \affiliation{College of Computing, Université Mohamed 6 Polytechnique, Benguerir, Morocco}

\date{\today}

\begin{abstract}
We study the homogenization of the Thomas--Fermi--von Weizs{\"a}cker (TFW) model for 2D materials introduced in Ref.~\onlinecite{BB}. It consists in considering 2D-periodic nuclear densities with periods going to zero. We study the behavior of the corresponding ground state electronic densities and ground state energies. The main result is that these three-dimensional problems converge to a limit model that is one dimensional, similar to the one proposed in Ref.~\onlinecite{hb}. We also illustrate this convergence with numerical simulations and estimate the converging rate for the ground state electronic densities and the ground state energies.
\end{abstract}


\maketitle

\tableofcontents

\section{Introduction and main results}

In recent years, two-dimensional (2D) materials have become an intensely active research field, marked by landmark discoveries and widespread \added{investigation~\cite{VdW,R5,
review-phosphorene,
review-MOS2}}, driven by their unique physical, electrical, chemical and optical properties, which significantly differ from their 3D counterparts~\cite{review-calcul-properties-2}. Electronic structure simulations are highly useful in the discovery of these properties and their tuning for potential applications~\cite{review-calcul-properties, review-calcul-properties-2, hetero}. Thus, the need for mathematical models and simulation algorithms tailored for 2D materials arises~\cite{review-computational, method-2D, method-3}.

Density Functional Theory (DFT) is one of the most widely used simulation tools in electronic structure calculations. It consists of describing the electrons by their density $\rho$ and the energy of the system by a functional of $\rho$.
A famous model in this class is the (orbital free) Thomas--Fermi--von Weizs{\"a}cker (TFW) model~\cite{rf3, vW}. From a mathematical point of view, an important result in the study of crystals is the thermodynamic limit problem. It consists of proving that when a finite cluster converges to some periodic perfect crystal, the corresponding ground state electronic density and ground state energy per unit volume converge to the periodic equivalent. This program has been carried out for the Thomas--Fermi--von Weizsäcker model for three-dimensional (3D) crystals in 
Ref.~\onlinecite{catto}
and for one- and two-dimensional (1D and 2D) crystals in~Ref.~\onlinecite{BB}. \added{Other properties of TFW-type models for crystals have been studied in the literature. The seminal work in~Ref.~\onlinecite{rf3} established the basic theory for atoms, molecules and crystals. Subsequent research addressed aspects such as the analysis of defects~\cite{KumarRamabathiran2024, Erlacher2011}, stochastic materials~\cite{Blanc2007}, potential symmetry breaking~\cite{Ricaud2018}, rigorous derivation of DFT models for 2D materials~\cite{hb}, geometric optimisation~\cite{Blanc2001}, and the emergence of periodic phenomena in large atoms~\cite{BjergSolovej2024}.}

In this paper, we investigate the homogenization of the TFW model for 2D crystals. Our goal is to find a homogeneous material equivalent to a 2D crystal in the limit when the lattice parameter goes to zero. This corresponds to putting more and more (normalized) nuclei in the unit cell, or equivalently looking at the crystal from further and further away. It turns out that the homogenized material can be described by a 1D model, in the same spirit as in~Ref.~\onlinecite{hb}, which allows reducing computational time and resources required to simulate a 2D material, at least in a zero-order approximation. 
 Our proof can be generalized if we substitute the  $\int \rho^{5/3}$ term in the kinetic energy by $\int \rho^p $ for some \replaced{ $p>\frac{3}{2}$ }{ $\frac{3}{2}< p$ }. Note that the strict convexity of the energy functional gives the uniqueness of the ground state, which plays an important role in the proof.

To our knowledge, the closest work to ours is Ref.~\onlinecite{Mon} 
where the authors used the 2D TFW model~\cite{BB} to derive macroscopic features of a crystal from the microscopic structure in the presence of an external electric field. The crystal is modeled in the band $\RR^2\times [-1,1]$ and micro--macro limit is taken when the ratio between the atomic spacing and the size of the crystal goes to zero.

The article is organized as follows. We start by recalling the Thomas--Fermi--von Weizs{\"a}cker model for finite systems in Section~\ref{sec:TFW-finite}, and for 2D crystals in Section~\ref{sec:TFW-2D}. In Section~\ref{sec:main-result}, we define the homogenization process, along with the limit problem, and state our main result, whose proof is detailed in Section~\ref{sec:homo}. Intermediate results about the 2D Coulomb interaction are presented in Section~\ref{sec:hartree}. Finally, numerical illustrations are gathered in Section~\ref{sec:numerics}.

\medskip

\paragraph*{\bf Acknowledgments}

The research leading to these results has received funding from OCP grant AS70 ``Towards phosphorene-based materials and devices''. Prof. S. Lahbabi thanks the CEREMADE for hosting her during the final writing of this article.

\subsection{Thomas--Fermi--von Weizs{\"a}cker model for finite systems}\label{sec:TFW-finite}
We present in this section the TFW model for finite systems. Let 
\replaced{$m \in L^1(\RR^3, \RR^+)$ }{$m\in L^1(\RR^3)$}  
be a finite nuclear charge density.  The state of the electrons is described by a non-negative electronic density $ \rho\in L^1(\RR^3)$. The energy functional is given by
\begin{equation}
	\E^m(\rho)=\int_{\mathbb{R}^3} \abs{\nabla \sqrt{\rho}}^2+\int_{\mathbb{R}^3} \rho^{5/3}+\frac12 D(\rho-m,\rho-m),
	\label{eq:energy}
\end{equation}
where the first two terms represent the kinetic energy of the electrons and $D(f,g)$ is the Coulomb interaction between charge densities $f$ and $g$ in the Coulomb space \replaced{$\C=\set{f\in \S'(\RR^3,\RR),\; \frac{\widehat{f}}{\av{\cdot}}\in L^2(\RR^3)}$}{$\C=\set{f\in \S'(\RR^3),\; \frac{\widehat{f}}{\av{\cdot}}\in L^2(\RR^3)}$}. It is defined by 
$$ D(f,g)= \int_{\mathbb{R}^3}\int_{\mathbb{R}^3} \dfrac{f(x)g(y)}{\abs{x-y}}\diff x \diff y=4\pi\int_{\RR^3}\frac{
\overline{\widehat{f}(k)}\widehat{g}(k)}{\av{k}^2} \diff k, $$
where $\widehat{f}(k)$ denotes the \replaced{Fourier transform of $f$ evaluated at $k$}{the $k-$th Fourier coefficient for $f$}. The ground state is given by the following minimization
problem
\begin{equation}
I^m=\inf \left\{ \E^m(\rho),\; \rho\geq 0,\; \sqrt{\rho}\in H^1(\mathbb{R}^3),\int_{\mathbb{R}^3}\rho=\int_{\mathbb{R}^3}m  \right\}.
\label{TFW-R}
\end{equation}
It is well known that problem \eqref{TFW-R} has a unique minimizer $\rho$ (see for instance Ref.~\onlinecite{Ben}). $u=\sqrt{\rho}$ is the unique solution of the corresponding Euler--Lagrange equation
\begin{equation}\label{3D_EL}
\begin{cases}
	-\Delta u+\frac53 u^{7/3}+u\Phi=\lambda u, \\
	-\Delta \Phi=4\pi (u^2-m),
\end{cases} 
\end{equation}
where $\displaystyle  \lambda\in \RR, \mbox{ and }\Phi=(u^2-m)*\frac{1}{\av{\cdot}}$ is the mean--field potential.  

\subsection{Thomas--Fermi--von Weizs{\"a}cker model for 2D crystals}\label{sec:TFW-2D}
We present in this section the TFW model for 2D crystals introduced in Ref.~\onlinecite{BB}. 2D crystals are characterized by a nuclear density $m$ that has the periodicity of a 2D lattice 
$\R=a_1\ZZ+a_2\ZZ$, where $(a_1,a_2)$ are two linearly independent vectors in $\RR^2$ (see Figure~\ref{Tux}); namely 
$$
m(x_1+k_1,x_2+k_2,x_3)=m(x_1,x_2,x_3),\quad \forall x\in \RR^3,\; \forall (k_1,k_2)\in \R. 
$$

\begin{figure}[!hbt]
\centering
\setlength{\unitlength}{900sp}%

\begingroup\makeatletter\ifx\SetFigFont\undefined%
\gdef\SetFigFont#1#2#3#4#5{%
  \reset@font\fontsize{#1}{#2pt}%
  \fontfamily{#3}\fontseries{#4}\fontshape{#5}%
  \selectfont}%
\fi\endgroup%

\begin{picture}(5000,13000)(0,-10000)

{\color[rgb]{0,0,1}\thinlines\put(  -1200,-6800){\circle*{350}}}%
{\color[rgb]{0,0,1}\put(500,-6800){\circle*{350}}}%
{\color[rgb]{0,0,1}\put(2200,-6800){\circle*{350}} }%
{\color[rgb]{0,0,1}\put(3900,-6800){\circle*{350}} }%
{\color[rgb]{0,0,1}\put(5600,-6800){\circle*{350}} }%
{\color[rgb]{0,0,1}\put(7300,-6800){\circle*{350}} }%

{\color[rgb]{0,0,1}\put(-200,-5100){\circle*{350}} }%
{\color[rgb]{0,0,1}\put(1500,-5100){\circle*{350}} }%
{\color[rgb]{0,0,1}\put(3200,-5100){\circle*{350}} }%
{\color[rgb]{0,0,1}\put(4900,-5100){\circle*{350}}  }%
{\color[rgb]{0,0,1}\put(6600,-5100){\circle*{350}} }%
{\color[rgb]{0,0,1}\put(8500,-5100){\circle*{350}} }%

{\color[rgb]{0,0,1}\put(800,-3400){\circle*{350}}}%
{\color[rgb]{0,0,1}\put(2500,-3400){\circle*{350}}}%
{\color[rgb]{0.80, 0.80, 1.0}\put(4200,-3400){\circle*{350}}}%
{\color[rgb]{0,0,1}\put(5900,-3400){\circle*{350}}}%
{\color[rgb]{0,0,1}\put(7600,-3400){\circle*{350}} }%
{\color[rgb]{0,0,1}\put(9500,-3400){\circle*{350}} }%

{\color[rgb]{0,0,1}\put(1800,-1700){\circle*{350}}}%
{\color[rgb]{0.80, 0.80, 1.0}\put(3500,-1700){\circle*{350}}}%
{\color[rgb]{0.80, 0.80, 1.0}\put(5200,-1700){\circle*{350}}}%
{\color[rgb]{0,0,1}\put(6900,-1700){\circle*{350}}}%
{\color[rgb]{0,0,1}\put(8600,-1700){\circle*{350}} }%
{\color[rgb]{0,0,1}\put(10500,-1700){\circle*{350}} }%

{\color[rgb]{0,0,1}\put(2800,00){\circle*{350}}}%
{\color[rgb]{0.80, 0.80, 1.0}\put(4500,00){\circle*{350}}}%
{\color[rgb]{0,0,1}\put(6200,00){\circle*{350}}}%
{\color[rgb]{0,0,1}\put(7900,00){\circle*{350}}}%
{\color[rgb]{0,0,1}\put(9600,00){\circle*{350}} }%
{\color[rgb]{0,0,1}\put(11500,00){\circle*{350}} }%

{\color[rgb]{1,0,0}\put(3550,-2550){\circle*{200}} }%
{\color[rgb]{1,0,0}\put(5250,-2550){\circle*{200}} }%
{\color[rgb]{1,0,0}\put(3150,-4250){\circle*{200}} }
{\color[rgb]{1,0,0}\put(4650,-4250){\circle*{200}} }%

\put(3300,-2400){$Q$}
\put(3300,-2550){\line( 1,0){1700}}
\put(2900,-4250){\line( 1,0){1700}}

\put(2900,-4250){\line(1,4){410}}
\put(4600,-4250){\line(1,4){410}}

\put(3800,-3400){\vector(-1,-2){2600}}
\put(400,-8500){$x_1$}

\put(4000,-3400){\vector(1,0){7000}}
\put(10000,-3100){$x_2$}

\put(4000,-3400){\vector(0,1){5000}}
\put(4300,1000){$x_3$}


\put(3300,2000){\line( 1,0){1700}}
\put(2900,500){\line( 1,0){1700}}
\put(2900,500){\line(1,4){370}}
\put(4600,500){\line(1,4){370}}


\put(3300,2000){\line( 0,-1){10000}}
\put(5000,2000){\line( 0,-1){10000}}
\put(2900,500){\line(0,-1){10000}}
\put(4600,500){\line(0,-1){10000}}

\put(3300,-8000){\line( 1,0){1700}}
\put(2900,-9500){\line( 1,0){1700}}
\put(2900,-9500){\line(1,4){370}}
\put(4600,-9500){\line(1,4){370}}

\put(5000,-9500){$\Gamma$}

\end{picture}%
\caption{Example of a 2D lattice and its unit cell $\Gamma$. }
\label{Tux}
\end{figure}
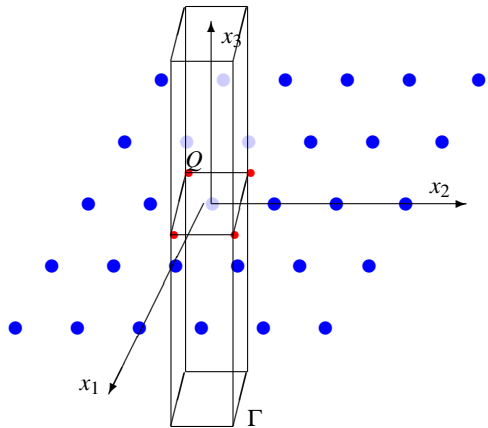

\noindent From now on, we denote by $Q$ the unit cell of $\R\subset\RR^2$ and by $\Gamma=Q\times \RR$ the unit cell of $\R$ seen as a lattice in $\RR^3$. For $x\in \RR^3$, we denote by $\underline{x}=(x_1,x_2)\in \RR^2$ so that $x=(\underline{x},x_3)$.

By means of a thermodynamic limit procedure, it has been shown in Ref.~\onlinecite{BB} that 2D crystals can be described by a model similar to \eqref{eq:energy}--\eqref{TFW-R} posed in the unit cell $\Gamma$. The main difference is that the 3D Green function $ \displaystyle \frac{1}{\av{x}}$ is replaced by the 2D periodic Green function $G$ solution of 
$$
-\Delta G=4\pi \sum_{k\in \R\times \set{0}}\delta_k.
$$
An explicit formula of $G$ is given by Ref.~\onlinecite[equation 11]{BB}
\begin{equation}
	G(x)=-\dfrac{2\pi}{\av{Q}}\abs{x_3}+ \sum_{\underline{k}\in \R}\left(
	\dfrac{1}{\abs{{x}-(\underline{k},0)}}-\dfrac{1}{\av{Q}}\int_{Q}  \dfrac{\diff \underline{y}}{\abs{{x}-\left(\underline{y}+\underline{k},0\right)}}
	\right).
	\label{eq:def G}
\end{equation}
It can be seen as the sum over the lattice of the Coulomb potential created by a point charge placed at the lattice sites, screened  by a uniform background of negative unit charge. A Fourier decomposition of $G$ can be found in Ref.~\onlinecite{hb}. The 2D crystals energy functional then reads 
\begin{equation}
	\displaystyle \E^m_{\per}(\rho) =\int_{\Gamma } \abs{\nabla \sqrt{\rho}}^2+\int_{\Gamma } \rho^{5/3}+\frac12 D_G(m-\rho,m-\rho),
	\label{eq:E_per}
\end{equation}
where the Hartree interaction $D_G$ is given by 
$$\displaystyle D_G(f,g):=\int_{\Gamma }\int_{\Gamma }{f(x)g(y)}G(x-y)\diff x\diff y.$$
\begin{Rem}
	In Ref.~\onlinecite{BB}, there is a Coulomb correction term in the energy which comes from the Dirichlet boundary condition at infinity considered in the thermodynamic limit procedure, so that the energy is 
	\begin{equation*}
		\displaystyle \widetilde{\E^m_{\per}}(\rho) =\E^m_{\per}(\rho)+\dfrac{1}{\av{Q}}\int_{\Gamma }\int_{\Gamma }\dfrac{ \text{\replaced{$m(x)m(y)$}{$f(x)g(y)$}}}{\abs{x-y}}. 
	\end{equation*}
	In our study, we omit this correction term as it does not affect the problem from a mathematical point of view. 
\end{Rem}
\noindent The ground state is given by the minimization problem
\begin{equation}
	I^m_{\per}=\inf \left\{ \E^m_{\per}(\rho),\quad \rho\geq 0,\quad \sqrt{\rho}\in X_{\per},\quad\int_{\Gamma}\rho=\int_{\Gamma}m \right\}
	\label{TFW-per}
\end{equation}
where 
$$X_{\per}=\left\{ v \in H^1_{\per}(\Gamma ),\; \left( 1+\abs{x_3} \right)^{1/2}v \in L^2_{\per}(\Gamma ) \right\}.$$

\added{The weight $(1+|x_3|)^{1/2}$ in the definition of the function space $X_{per}$ is there to insure that $\rho:=v^2$, with $v\in X_\per$, satisfies $(1+\av{x_3})\rho\in L^1(\Gamma)$ and thus has a finite interaction energy, as  the periodic potential $G$ contains a term proportional to $|x_3|$. 
} This problem has been studied in~\added{Ref.~\onlinecite[section~3.2.3]{BB}}, along with some basic properties of the solution, that we summarize in the following Theorem. 

\begin{Th}[{Ref.~\onlinecite[Theorem 3.2]{BB}}]
Let $m\neq 0$ be a smooth non-negative, $\R$-periodic function with compact support w.r.t $x_3$. 
Then, the minimization problem~\eqref{TFW-per} has a unique minimizer $\rho$. \added{Moreover,} $u=\sqrt{\rho}$ is the unique solution of the corresponding Euler--Lagrange system 
\begin{equation}
    \begin{cases}
    -\Delta u+\frac{5}{3}u^{7/3}+u\Phi=\lambda u,\\
    -\Delta \Phi=4\pi(u^2-m),\\
    \end{cases}
    \label{base}
\end{equation}
where $\lambda \in\RR$. In addition, $u\in L^{\infty}(\mathbb{R}^3)$ and $\abs{u(x)}\leq \dfrac{C}{1+\abs{x_3}^{3/2}}$, for $\abs{x_3}>1$, $C>0$ being a constant independent of the density $m$.
\label{Th2}
\end{Th}

\begin{Rem}
The fact that, in inequality $ \textstyle \av{u(x)}\leq  \dfrac{C}{1+\abs{x_3}^{3/2}} $, for $\abs{x_3}>1$, the constant is independent of the density $m$,  comes from a supersolution method applied to equation \eqref{base} (see Ref.~\onlinecite[Theorem 3.1 \replaced{and}{} Theorem 2.3]{BB}).
\end{Rem}

\subsection{Homogenization of 2D materials}\label{sec:main-result}
In the framework of the Thomas--Fermi (TF) model and the reduced Hartree--Fock model (rHF), the recent work Ref.~\onlinecite{hb} studies reduced models for 2D homogeneous materials. The idea is that if the material is homogeneous, the 3D model is equivalent to a 1D model; simpler and less costly to simulate. In the present work, we are interested in the homogenization procedure, which models looking at a crystal macroscopically, from further and further away. Namely, we put more and more nuclei in the unit cell, with the right charge normalization, and  we ask the following questions:
\begin{itemize}
	\item What is the nuclear density at the limit?
	\item Does the ground state electronic density, potential and energy converge?
	\item What is the model describing the limit electronic structure?  
\end{itemize}

Let $m$ be a nuclear density satisfying the conditions of Theorem~\ref{Th2} and consider the  following sequence of nuclear densities which consists in putting \replaced{$N^2$}{$N$} small nuclei in the unit cell  
\begin{equation}\label{eq:mN}
m_N(x_1,x_2,x_3):=m(Nx_1,Nx_2,x_3),\quad \forall x\in \Gamma.
\end{equation}
We note that $m_N$ is $\frac{1}{N}\R$-periodic, which describes a more homogeneous material than the initial one. The limit when $N \rightarrow +\infty$ describes a homogeneous material (see Figure~\ref{com}).
\begin{figure}[!hbt]
\centering
\includegraphics[width=120mm,scale=1]{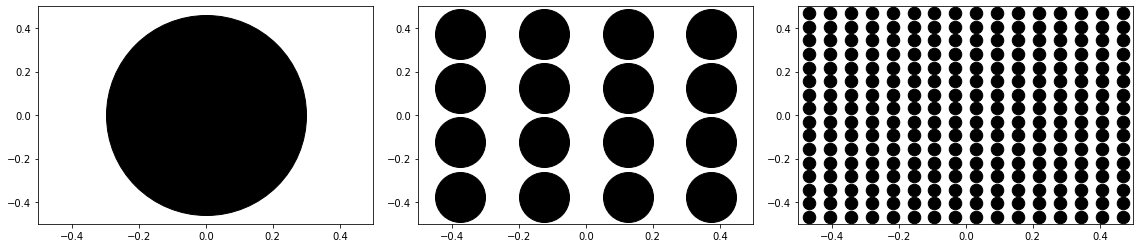}
\caption{The homogenization process for the nuclear density $m_N$, illustrated at $N=1,4$ and $16$ from left to right.}
\label{com}
\end{figure}
We will show that this model \replaced{\enquote{converges}}{"converges"}, in some sense that we will precise later, to the following limit model. For a 1D nuclear density $\mu$, we introduce the energy functional 
\begin{equation}\label{eq:Emu1}
	\E_1^\mu(\rho)=\int_{\mathbb{R}}\abs{\nabla \sqrt{\rho}}^2 +\int_{\mathbb{R}}\rho^{5/3}+\frac12 D_1(\rho-\mu, \rho-\mu),
\end{equation}
where the 1D Hartree interaction $D_1$ is defined for functions decaying fast enough by 
$$D_1(f,g)=-{2\pi}\int_{\mathbb{R}} \int_{\mathbb{R}}\abs{s-t}f(s)g(t)\diff t \diff s.$$
\added{Here, the 1D Green kernel $ -2 \pi \abs{t}$ is used}. The ground state of this model is given by the minimization problem
\begin{equation}
    I^\mu_1= \min\{E^\mu(\rho),\; \rho \geq 0,\; \sqrt{\rho} \in X_1, \int_{\mathbb{R}}\rho=\int_\RR \mu\},
    \label{TFW-lim}
\end{equation} 
where
$$ X_1=\left\{ v \in H^1(\mathbb{R}), \;(1+\abs{\cdot})^{1/2}v\in L^2(\mathbb{R}) \right\}. $$ 

\noindent The following Theorem, which is a direct consequence of Theorem~\ref{Th2} and the definitions of Hartree interaction terms $D_G$ and $D_1$, is similar to Ref.~\onlinecite[Theorems 2.2 and 2.8]{hb}, which treat Thomas--Fermi and reduced Hartree--Fock models.  
\begin{Th}
Let $\mu \in \D(\RR)$. The problem \eqref{TFW-lim} has a unique minimizer. Moreover, $u:=\sqrt{\rho}$ is the unique solution of the corresponding Euler--Lagrange equations 
\begin{equation*}
    \begin{cases}
    - u''+\frac53 u^{7/3}+u\Phi=\lambda u,\\
    - \Phi''=4\pi(u^2-\mu), 
    \end{cases}
\end{equation*}
where $\lambda\in\RR$, 
and there exists $C>0$, independent of $\mu$, such that for any  $\av{t}\geq 1$, $\av{u}\leq \frac{C}{1+ \av{t}^{3/2}}$. 
\label{thf}
\end{Th}
\noindent  Our main contribution is the following Theorem. 
\begin{Th}
Let $m$ be a nuclear density satisfying the hypothesis of Theorem~\ref{Th2} and for $N\in\NN\setminus \set{0}$, let $m_N$ be defined as in~\eqref{eq:mN}. We denote by \replaced{$I_N=I^{m_N}_{\per}$}{$I_N=I^{m_N}$} and by $\rho_N$ 
the corresponding ground state given by Theorem~\ref{Th2}.
Let $m_0(t)=\frac{1}{\av{Q}}\int_{Q}m(\underline{x},t)\diff \underline{x}$ and denote by $I_0=I^{m_0}_1$ and by $\rho_0$ 
 the corresponding ground state given by Theorem~\ref{thf}.
  The following holds
\begin{enumerate}[i.]
    \item \label{i} $\displaystyle \lim_{N \rightarrow \infty} I_N = I_0$,
    \item \label{ii} $\rho_N$ converges to $\rho_{0}$ in $L^1(\Gamma)$, in $L^p_{loc}(\Gamma)$ for all $ 1\leq p \leq 3$, and almost everywhere,
    \item \label{iii} $\sqrt{\rho_N}$ converges \added{to} $ \sqrt{\rho_0}$ weakly in $H^1_\per(\Gamma)$. 
 \end{enumerate}
\label{Th-H}
\end{Th}

Our result is a zero order approximation of a finite $N$ situation. In homogenization theory of partial differential equations with  periodically oscillating coefficients, the two scale convergence technique is usually used to find higher order terms in the approximation~\cite{allaire}. Applying a similar approach to our problem is a perspective of this work. 

\section{Hartree interaction} 
\label{sec:hartree}
We recall in this section some properties of the Green function $G$, the Hartree interaction $D_G$ and prove a uniqueness result of the mean field potential $\Phi$.
In the subsequent proposition, we introduce a convenient decomposition of the kernel $G$, which proves to be useful throughout the paper.

\begin{Pro}\label{Deco-G}
	We have 
$$G(x)=-\dfrac{2\pi}{\av{Q}}\abs{x_3}+ \dfrac{1}{\abs{x}}+\psi(x),$$
where $\psi\in L^\ii(\Gamma)$.
\end{Pro}

\noindent Proposition~\ref{Deco-G} is partially proved in Ref.~\onlinecite{BB}. We detail the proof here for consistency. 

\begin{proof}

\added{We start from the formula~\eqref{eq:def G} of the explicit expression of the  periodic Green function $G$ 
\begin{equation}
	G(x)=-\dfrac{2\pi}{\av{Q}}\abs{x_3}+ \sum_{\underline{k}\in \mathcal{R}}\left(
	\dfrac{1}{\abs{{x}-(\underline{k},0)}}-\dfrac{1}{\av{Q}}\int_{Q}  \dfrac{\diff \underline{y}}{\abs{{x}-\left(\underline{y}+\underline{k},0\right)}}
	\right).
	\label{eq:def G proof direct}
\end{equation}
We then isolate the term corresponding to $\underline{k}=0$ from the sum
\begin{align*}
G(x) = -\dfrac{2\pi}{\av{Q}}\abs{x_3} &+ \left( \dfrac{1}{\abs{x-(\underline{0},0)}} - \dfrac{1}{\av{Q}}\int_{Q}  \dfrac{\diff \underline{y}}{\abs{{x}-\left(\underline{y}+\underline{0},0\right)}} \right)\\ & \qquad \qquad \qquad \qquad + \sum_{\underline{k}\in \mathcal{R}^*}\left(	\dfrac{1}{\abs{{x}-(\underline{k},0)}}-\dfrac{1}{\av{Q}}\int_{Q}  \dfrac{\diff \underline{y}}{\abs{{x}-\left(\underline{y}+\underline{k},0\right)}}\right).
\end{align*}
We can thus rewrite $G$ as
\begin{align*}
 G(x) = -\dfrac{2\pi}{\av{Q}}\abs{x_3} + \dfrac{1}{\abs{x}} & - \dfrac{1}{\av{Q}}\int_{Q}  \dfrac{\diff \underline{y}}{\abs{{x}-(\underline{y},0)}} \\
 & \qquad \qquad + \sum_{\underline{k}\in \mathcal{R}^*}\left(
	\dfrac{1}{\abs{{x}-(\underline{k},0)}}-\dfrac{1}{\av{Q}}\int_{Q}  \dfrac{\diff \underline{y}}{\abs{{x}-(\underline{y}+\underline{k},0)}}
	\right).
\end{align*}
Therefore we can write $G(x) = -\frac{2\pi}{|Q|}|x_3| + \frac{1}{|x|} + \psi(x)$, where $\psi(x)$ is the sum of the remaining terms }
$$
\psi(x) = -\dfrac{1}{\av{Q}}\int_{Q}  \dfrac{\diff \underline{y}}{\abs{{x}-(\underline{y},0)}} + \sum_{\underline{k}\in \mathcal{R}^*}\left(
	\dfrac{1}{\abs{{x}-(\underline{k},0)}}-\dfrac{1}{\av{Q}}\int_{Q}  \dfrac{\diff \underline{y}}{\abs{{x}-(\underline{y}+\underline{k},0)}}
	\right).
$$
Let us show that $\psi\in L^\ii(\Gamma)$. The function $\displaystyle x\mapsto \dfrac{1}{\av{Q}}\int_{Q}  \dfrac{\diff \underline{y}}{\abs{{x}-(\underline{y},0)}}$ is continuous on $\Gamma$  and goes to zero as $\av{x_3}\rightarrow \infty$. Moreover, the sum defining $\psi$ converges normally on a neighborhood of $0_{\RR^3}$ (see Ref.~\onlinecite[Proposition 3.2]{BB}). Therefore it is continuous on that neighborhood, which implies that $\psi$ is bounded on $Q\times [-\epsilon,\epsilon]$, for some $\epsilon>0$. From Ref.~\onlinecite[Equation 3.12]{BB}, there exists $C>0$ such that 
	$$ \sum_{\underline{k}\in (\R)^*}\left( 
	\dfrac{1}{\abs{{x}-(\underline{k},0)}}-\dfrac{1}{\av{Q}}\int_{Q}  \dfrac{\diff \underline{y}}{\abs{{x}-(\underline{y}+\underline{k},0)}}
	\right)\leq \dfrac{C}{\av{x_3}^\alpha},\quad \forall\alpha<1.$$
	It follows that \deleted{the} $x\mapsto  \sum_{\underline{k}\in (\R)^*}\left( 
	\dfrac{1}{\abs{{x}-(\underline{k},0)}}-\dfrac{1}{\av{Q}}\int_{Q}  \dfrac{\diff \underline{y}}{\abs{{x}-(\underline{y}+\underline{k},0)}}
	\right) $ is bounded on $x\in Q\times (\RR\setminus  [-\epsilon,\epsilon] )$, thus $\psi$ is also bounded on the same interval.
\end{proof}
As a consequence of the above proposition, we prove useful properties of the potential $\Phi$. Let us introduce some notations. For some domain $\Omega\subset  \RR^d$ and $1\leq p\leq \ii$, we  denote by 
$$L^p_{\rm unif}(\Omega)=\left\{f \in L^p_{\rm loc}(\Omega),\;\forall r>0,\; \sup_{ \text{\replaced{$B(x,r)$}{$B_r$}}\subset\Omega} \norm{f}_{L^p(B(x,r))} < \infty\right\},$$
where $B(x,r)$ is the ball of radius $r$ centered at $x$. For $f\in L^p(\Gamma)$ and $g\in L^q_{\rm per}
(\Gamma)$, the convolution
$$
(g*_\Gamma f)(x)=\int_\Gamma f(y)g(x-y) \diff y
$$
is well-defined in $L^r_{\rm per}(\Gamma)$, where $\frac1p+\frac1q=1+\frac1r$.  
\begin{Pro}
	\label{lemphi}
	Let 
		$$Y_{\per}=\left\{f \in L^1(\Gamma )\cap L^{5/3}(\Gamma),\; \int_{\Gamma} f=0,\; \abs{\cdot} f \in L^1(\Gamma ) \right\}.$$
	The map $$\begin{array}{rll}
		Y_\per & \to & L^\ii(\Gamma) \\
		f&\mapsto & G*_\Gamma f
	\end{array}
	$$
	is \replaced{well-defined}{well defined} and continuous.  Moreover, for $f\in Y_\per$, $G*_\Gamma f$ is the unique solution, up to an additive constant, of the Poisson equation
	\begin{equation}
		\left\{
		\begin{array}{rll}
			\Phi & \in & L^1_\unif(\RR^3)\\
			-\Delta \Phi& 	= & 4\pi f.
		\end{array}
		\right.
		\label{li}
	\end{equation}
\end{Pro}
\begin{proof}
 
From Proposition~\ref{Deco-G}, we have $G(x)=\dfrac{1}{\abs{x}}-\dfrac{2\pi}{\av{Q}}\abs{x_3}+\psi(x)$, with $\psi\in L^\infty(\Gamma)$. 
We are thus going to bound the three functions 
	\begin{equation*}
	f*_{\Gamma } \abs{x_3},\; 
		f*_{\Gamma } \frac{1}{\abs{x}} \; \text{and}\;
			f*_{\Gamma }\psi
	\end{equation*}
	in $L^{\infty}(\Gamma )$. First, since $\psi\in L^\infty(\Gamma)$, we have for any $x\in \Gamma$ 
	$$
	\av{\int_\Gamma f(x-y) \psi(y)\diff y}\leq \norm{\psi}_{L^\infty(\Gamma)}\norm{f}_{L^1(x-\Gamma)}.
	$$
	Thus
	$$
	\norm{ f*_{\Gamma }\psi}_{L^\infty(\Gamma)}\leq \norm{\psi}_{L^\infty(\Gamma)}\norm{f}_{L^1(\Gamma)}.
	$$
	We move to ${f}*_{\Gamma }\frac{1}{\abs{x}}$. We have
	\begin{align*}
		\av{f*_\Gamma \frac{1}{\abs{x}}(x)} & \leq \int_{\Gamma }\frac{\abs{f(x-y)}}{\abs{y}}\diff y
		\displaystyle \leq \int_{\Gamma }\abs{f(x-y)}\diff y+ \int_{\Gamma }\frac{\abs{f(x-y)}}{\abs{y}}\mathbbm{1}_{\abs{y}<1}\diff y. 
	\end{align*}
 For the second term, we use H\"{o}lder inequality with $f \in L^{5/3}(\Gamma )$ and $\frac{1}{\abs{y}}\mathbbm{1}_{\abs{y}<1} \in L^{5/2}(\Gamma )$. We conclude that
	\begin{equation*}
		\begin{array}{ll}
			\norm{\frac{1}{\abs{x}}*_{\Gamma }f}_{L^\infty} &\displaystyle \leq \norm{f}_{L^{1}(\Gamma )}+ \norm{f}_{L^{5/3}(\Gamma )}\norm{\frac{1}{\abs{y}}\mathbbm{1}_{\abs{y}<1}}_{L^{5/2}(\Gamma )}. 
		\end{array}
	\end{equation*}
	Regarding the last term, we use the neutrality assumption to write
	$$ \int_{\Gamma } f(x-y)\abs{y_3}\diff y=\int_{\Gamma } f(x-y)(\abs{y_3}-\abs{x_3})\diff y.$$
	Thus, using the triangle inequality and the $\mathcal{R}$-periodicity of $x\mapsto x_3f(x)$ , we obtain
	\begin{align*}
	\av{ \int_{\Gamma } f(x-y)\abs{y_3}\diff y }\leq \int_{\Gamma }\av{ f(x-y)}\abs{x_3-y_3}\diff y =\int_{\Gamma } \av{f(y)}\abs{y_3}\diff y .
	\end{align*}
%
%
\added{We now prove that  $ G*_{\Gamma }f $ is the unique solution of~\eqref{li}. The function $G*_{\Gamma }f$ is  $\mathcal{R}$-periodic, therefore it is bounded over $\RR^3$ and  also in $L^1_\unif(\RR^3)$. It satisfies equation ~\eqref{li} by the definition of $G$. Let $\Phi$ be  another solution. Then  $h=\Phi-G*_{\Gamma }f\in L^1_\unif(\RR^3)$ is a harmonic function over $\mathbb{R}^3$.}  By the mean value Theorem for \deleted{a} harmonic functions, we have that \added{for any $x \in \RR^3$,
	$$ h(x)=\frac3{4\pi}\int_{B(x,1)} h(y)\diff y.$$
    Therefore,
	$$ |h(x)| \leq \frac3{4\pi} \sup_{z\in \RR^3} \int_{B(z,1)}|h(y)|\diff y. $$
	Since $h \in L^1_\unif(\RR^3)$, the supremum is finite, which implies $h$ is bounded over $\RR^3$.} By Liouville's Theorem, we conclude that $h$ is constant. 
\end{proof}

\begin{Cor}\label{D_G bound} Let $f\in Y_{per}$, then $\nabla (G*_\Gamma f)\in \bra{L^2(\Gamma)}^3$ and there exists  $C>0$, such that
	$$D_G(f,f)= \added{\text{$(4\pi)^{-1}$}} \norm{\nabla (G*_\Gamma f)}^2_{L^2(\Gamma)}\leq C \norm{f}_{L^{1}(\Gamma)}\left(\norm{f}_{L^{1}(\Gamma)}+\norm{f}_{L^{5/3}(\Gamma)}+\norm{|x|f}_{L^{1}(\Gamma)}\right).$$
\end{Cor}

\begin{proof} We have 
		\begin{align*}
	D_G(f,f)= \int_\Gamma (G*_\Gamma f)\times f=\added{\text{$-(4\pi)^{-1}$}}\int_\Gamma (G*_\Gamma f)\times \Delta (G*_\Gamma f)
	\end{align*} 
    

\added{In Ref.~\onlinecite[Section 3.2.1]{BB}, it is shown using a hybrid Fourier transform, that is discrete in $\underline{x}$ and continuous in $x_3$, that
$$
D_G(f,f) = \sum_{k \in \R} \int_{R} \frac{\left|\widehat{f}(k,\xi)\right|^2}{ \pi (|k|^2 + \xi^2 )}  d\xi.
$$
Hence, 
$$
D_G(f,f) = (4\pi)^{-1} \sum_{k \in \R} \int_{R} \left| \frac{ (i 2 \pi (k,\xi) ) \widehat{f}(k,\xi) }{ \pi (|k|^2 + \xi^2 )} \right|^2 d\xi = (4\pi)^{-1}\int_{\Gamma} | \widehat{\nabla(G *_{\Gamma} f )}| ^2 d x .
$$
By the Planchrel-Parseval property\cite{BB} of the hybrid Fourier transform, we conclude that
\begin{align*}
	D_G(f,f)=\text{$(4\pi)^{-1}$} \int_\Gamma \av{\nabla(G*_\Gamma f)}^2.
\end{align*} 
}
Besides, by the previous proposition $ G*_{\Gamma }f \in L^{\infty}(\Gamma)$ and 
$$\norm{G*_{\Gamma }f}_{L^{\infty}(\Gamma)}\leq (\norm{\psi}_{L^\ii(\Gamma)}+1)\norm{f}_{L^1(\Gamma)}+ \norm{f}_{L^{5/3}(\Gamma )}\norm{\frac{1}{\abs{y}}\mathbbm{1}_{\abs{y}<1}}_{L^{5/2}(\Gamma )}+\norm{\abs{x}f}_{L^{1}(\Gamma )}\leq C\norm{f}_{Y_{ \per}},$$
 which proves the inequality stated in the corollary.

\end{proof}

\section{Homogenization procedure: proof of Theorem~\ref{Th-H} }\label{sec:homo}
This section is devoted to the proof of Theorem~\ref{Th-H}. The strategy of the proof is as follows. We start by proving the convergence of the nuclear densities \replaced{$(m_N)_N$}{$(m_N)$}.  Indeed, when $N$ increases, the nuclear density $m_N$ becomes more homogeneous and the sequence \replaced{$(m_N)_N$}{$(m_N)$} converges to the 2D homogeneous density $m_0(x_3)=\frac{1}{\av{Q}} \int_Q m(x)\diff \underline{x}$ (see Lemma~\ref{pro:cv-mN}). Electronic densities $\rho_N$ are uniformly bounded \replaced{w.r.t.}{with respect to} $N$. We can thus extract convergent subsequences (see Proposition~\ref{prop:cv-rhoN}). Both convergences give the upper bound \replaced{$I_0 \leq \liminf I_N$}{$I \leq \liminf I_N$}. 
To prove the lower bound $I_0 \geq \limsup I_N$, we use $\rho_0$ as a test function. The proof
is divided into four steps. 

\smallskip
\paragraph*{Step 1: Properties of the sequence \replaced{$(m_N)_N$}{$(m_N)$}.}

We state below two properties of the sequence \replaced{$(m_N)_N$}{$(m_N)$}. 

\begin{Lem}\label{Neutrality} 
For $p\geq 1$ and $f\in L^1_\loc(\RR)$  such that $x\mapsto f(x_3)m^p(x)\in L^1(\Gamma)$, we have
$$\displaystyle\int_Q m_N^p(\underline{x},x_3)\diff \underline{x}= \int_Q m^p(\underline{x},x_3)\diff \underline{x}$$
and 
	  $$\displaystyle\int_\Gamma f(x_3)m_N^p(x)\diff x=\int_\Gamma f(x_3)m^p(x)\diff x.$$
\end{Lem}

\begin{proof}
	We have
	$$\int_{Q} m_N^p(x_1,x_2,x_3)\diff x_1 \diff x_2 =	\int_{Q} m^p(Nx_1,Nx_2,x_3)\diff x_1 \diff x_2= \frac{1}{N^2}\int_{NQ} m^p(y_1,y_2,x_3)\diff y_1 \diff y_2.$$
%
As $m$ is $Q$-periodic, then $\displaystyle\int_{NQ} m^p(y_1,y_2,x_3)\diff y_1 \diff y_2=N^2\int_{Q} m^p(y_1,y_2,x_3)\diff y_1 \diff y_2$. Hence, the first claim is proved.  The second claim easily follows.
\end{proof}

\begin{Rem} 
By the second point of Lemma~\ref{Neutrality}, with $f=1$ and $p=1$, we have that $\int_\Gamma m_N=\int_\Gamma m= \int_\Gamma m_0$, so that any admissible state for $I_N$ is also an admissible state for $I_0$, and vice versa. 
\end{Rem}

\begin{Lem}
	The sequence \replaced{$(m_N)_N$}{$(m_N)$}  converges to $m_0$ weakly in $L^p(\Gamma)$ \replaced{for all }{$\forall$} $ 1\leq p < +\ii$. 
	\label{pro:cv-mN}
\end{Lem}

\begin{proof}
From the $Q$-periodicity of $m$, we have for any $x_3\in \RR$ (see for instance Ref.~\onlinecite{Luk})
\begin{equation}\label{eq:w-cv}
m_N(\cdot,\cdot,x_3) \rightharpoonup \added{ m_0(x_3) }
        \text{\deleted{ $m_0(\cdot,\cdot,x_3)$ }} 
		\quad \mbox{ in } L^p(Q), \quad \text{for all  } 1\leq p< \ii.
\end{equation}
\replaced{For $\phi\in L^q(\Gamma)$}{$\phi\in L^q(Q)$ and $\psi\in L^q(\RR)$}, with $\frac{1}{q}+\frac{1}{p}=1$, we have
\deleted{$f_N(x_3):=\displaystyle\psi(x_3)\int_Q  m_N(x)\phi(\underline{x})\diff  \underline{x}\to f(x_3):= m_0(x_3)\psi(x_3)\int_Q \phi(\underline{x})\diff \underline{x}$ a.e. and, }
\added{$f_N(x_3):=\displaystyle\int_Q  m_N(x)\phi(\underline{x}, x_3)\diff  \underline{x}\to f(x_3):= m_0(x_3)\int_Q \phi(\underline{x}, x_3)\diff \underline{x}$ a.e. and, by Hölder's inequality, we obtain
}

\deleted{
$\hspace{2cm}
    \av{f_N(x_3)} &\leq \norm{\phi}_{L^q(Q)}\av{\psi(x_3)} \bra{\int_Q \av{m_N(\underline{x},x_3)}^p\diff \underline{x}}^{1/p}$
    \\
    $
    \hspace{3cm}&= \norm{\phi}_{L^q(Q)}\av{\psi(x_3)} \bra{\int_Q\av{m(\underline{x},x_3)}^p\diff \underline{x}}^{1/p} \in L^1(\mathbb{R}).
$
}

\added{
\begin{align}
\begin{split} \label{eq:domination}
    \av{f_N(x_3)} &\leq \norm{\phi(\cdot,\cdot,x_3)}_{L^q(Q)} \bra{\int_Q \av{m_N(\underline{x},x_3)}^p\diff \underline{x}}^{1/p} \\
                 &= \norm{\phi(\cdot,\cdot,x_3)}_{L^q(Q)}\bra{\int_Q\av{m(\underline{x},x_3)}^p\diff \underline{x}}^{1/p} \in L^1(\mathbb{R}).
\end{split}
\end{align}
}

Thus, by the dominated convergence theorem, 

\deleted{\text{
$ \int_\RR f_N(x_3)\diff x_3=\int_\Gamma m_N(x)\phi(\underline{x})\psi(x_3)\diff x \to \int_\RR f(x_3)\diff x_3= \int_\Gamma m_0(x_3)\phi(\underline{x})\psi(x_3)\diff x.$
}}

\added{	$$ \int_\RR f_N(x_3)\diff x_3=\int_\Gamma m_N(x)\phi(\underline{x}, x_3)\diff x \to \int_\RR f(x_3)\diff x_3= \int_\Gamma m_0(x_3)\phi(\underline{x}, x_3)\diff x. 
$$}

\deleted{By the density of $L^q(Q)\otimes L^q(\RR)$ in $L^q(\Gamma)$, the proof is complete.}

\end{proof}

\paragraph*{Step 2: Convergence of electronic densities.}
The following proposition gives uniform bounds on  quantities of interest \replaced{w.r.t.}{with respect to} $N$.

\begin{Pro}\label{apriori}
There exist various constants $C$ independent of $N$ such that the following bounds hold.
\begin{enumerate}[i.]
    \item\label{II} $\displaystyle I_N\leq C$,
    \item\label{III} $\displaystyle\norm{ \sqrt{\rho_N}}_{H^1(\Gamma)} \leq C$,
    \item\label{IV} $\displaystyle  \norm{\rho_N}_{L^p({\Gamma })} \leq C $ for all $1\leq p \leq 3$, 
    \item\label{V} $ \displaystyle D_G(\rho_N-m_N,\rho_N-m_N) \leq C $.
\end{enumerate}
\end{Pro}

\begin{proof}
As $\rho_N$ is the minimizer of $I_N$ and since $m$ is an admissible test function for $\E^{m_N}_\per$, we have 
$$I_N=\E_{per}^{m_N}(\rho_N)\leq \E_{per}^{m_N}(m)=\int_{\Gamma } \abs{\nabla \sqrt{m}}^2+\int_{\Gamma } m^{5/3}+\frac12 D_G(m-m_N,m-m_N).$$
By Corollary \ref{D_G bound} applied to $f=m-m_N$ and Lemma~\ref{Neutrality}, we have 
$$
\begin{array}{lll}
	D_G(m-m_N,m-m_N) &\leq & C\norm{m-m_N}_{L^{1}(\Gamma)}\left(\norm{m-m_N}_{L^{1}(\Gamma)}+\norm{m-m_N}_{L^{5/3}(\Gamma)} \right. \\ 
	& & \left. +\norm{|x|(m-m_N)}_{L^{1}(\Gamma)}\right) \\
	& \leq & C\norm{m}_{L^{1}(\Gamma)}\left(\norm{m}_{L^{1}(\Gamma)}+\norm{m}_{L^{5/3}(\Gamma)}+\norm{|x|m}_{L^{1}(\Gamma)}\right), 
\end{array}$$
\added{where, $C$ denotes a generic positive constant that may differ from one line to the other.} \added{The finiteness of $\norm{\nabla \sqrt{m}}_{L^{2}(\Gamma)}$, $\norm{m}_{L^{5/3}(\Gamma)}$ and $\norm{|x|m}_{L^{1}(\Gamma)}$ follows directly from the assumed regularity and compact support of $m$ in $x_3$.} Points~{\it \ref{III}--\ref{V}} are easily deduced from \deleted{point} point~{\it\ref{II}}.
\end{proof}

\begin{Pro}
There exists a non-negative function $\rho_0\in L^1(\Gamma)\cap L^{5/3}(\Gamma)$ such that, up to a subsequence,
\begin{itemize}
\item  \replaced{$(\rho_N)_N$}{$(\rho_N)$} converges to  $\rho_0$ strongly in $L^1(\Gamma)$ and $L^p_{loc}(\Gamma ) $ for all $1\leq  p \leq 3$,  weakly in $L^p(\Gamma)$,  and almost everywhere on $\mathbb{R}^3$,
\item $(\nabla \sqrt{\rho_N})$ converges to $\nabla \sqrt{\rho_{0}}$ weakly in $(L^2(\Gamma ))^3$
\item $\rho_0-m_0\in Y_\per$ 
\item The sequence $\nabla\Phi_N=\nabla\left(G*_\Gamma (\rho_N-m_N)\right)$ converges to $\nabla\Phi_0=\nabla\left(G*_\Gamma (\rho_0-m_0)\right)$ weakly in $(L^2(\Gamma ))^3$.  
\end{itemize}
\label{prop:cv-rhoN}
As a consequence
$$
\E_\per^{m_0}(\rho_0) \leq \liminf I_N.
$$
\end{Pro}

\begin{proof}
The sequence $(\sqrt{\rho_N})_N$  is uniformly bounded \replaced{w.r.t.}{with respect to} $N$ in $H^1(\Gamma )$ by  Proposition~\ref{apriori}. Then, up to a subsequence, it converges to some non-negative function $u_0$ weakly in $H^1(\Gamma)$, strongly in $L^p_{loc}(\Gamma )$ for all $2\leq p\leq 6$ and almost everywhere on $\mathbb{R}^3$.  
Besides, by Theorem~\ref{Th2}, there exists $C\geq 0$ such that for any $N$ and any $x\in\RR^3$ such that $\av{x_3}\geq 1$, it holds 
$$
\av{u_N(x)}\leq \frac{C}{1+ \av{x_3}^{3/2}}. 
$$
By the almost everywhere convergence, $u_0$ satisfies the same estimate, and we have for $\rho_0:=u_0^2$ the following estimate
$$
\av{x_3}\rho_0(x)\leq \rho_0(x)1_{\av{x_3}\leq 1}+ \frac{C}{1+ \av{x_3}^3} \in L^1(\Gamma).
$$
Thus $\rho_0\in L^1(\Gamma)\cap L^{5/3}(\Gamma)$ and $\av{x_3}\rho_0\in L^1(\Gamma)$. 

We show now that $\rho_N$ converges to $\rho_0$ strongly in $L^1(\Gamma)$. 	Let $\epsilon>0$ and $R$ large enough such that for any $M,N\in \NN \setminus \{0\}$
\begin{align*}
		\int_{\av{x_3}>R}\av{\rho_N-\rho_M}\leq \frac{1}{R}\int_{\av{x_3}>R}\av{x_3}\av{\rho_N-\rho_M}\leq 	\frac{C}{R} \leq \frac{\epsilon}{2}.
	\end{align*} 
By the strong convergence of $\rho_N$ in $L^1_\loc(\Gamma)$, for $M,N$ large enough, we have 
$$
\int_{\av{x_3}<R}\av{\rho_N-\rho_M}\leq \frac{\epsilon}{2},
$$
which proves that \replaced{$(\rho_N)_N$}{$(\rho_N)$} is a Cauchy sequence in $L^1(\Gamma)$. Besides, \replaced{$(\nabla \Phi_N)_N$}{$(\nabla \Phi_N)$} is uniformly bounded in $L^2(\Gamma)$. Thus, it weakly converges, up to a subsequence, to $W\in \bra{L^2(\Gamma)}^3$. Since  $-\Delta \Phi_N =4\pi\left(\rho_N-m_N\right)$, then  \deleted{, we have} 
\begin{equation}\label{eq:w-cv-nabla-phi}
	\forall \phi \in \D(\Gamma), \quad  \int_\Gamma \nabla\Phi_N \cdot \nabla\phi =4\pi \int_\Gamma \left(\rho_N-m_N \right)\phi. 
\end{equation} 
By definition of the weak limit, the LHS of~\eqref{eq:w-cv-nabla-phi} converges to \replaced{$\int_\Gamma W \cdot \nabla \phi$}{$\int_\Gamma W\nabla \phi$}, and since $\rho_N-m_N$ weakly converges to $\rho_0-m_0$ in $L^1(\Gamma)$, then
\begin{equation*}
	\forall \phi \in \D(\Gamma), \quad \int_\Gamma W \cdot \nabla\phi =4\pi \int_\Gamma \left(\rho_0-m_0 \right)\phi.
\end{equation*}
Therefore $\div W=-4\pi (\rho_0-m_0)$ \added{in the sense of distributions. To conclude that $\int_\Gamma (\rho_0 - m_0) dx = 0$, we cannot directly apply the divergence theorem that needs the decay of $W$ at $x=\pm \infty$, which is, a priori, only $L^2$}. Instead, we use  the hybrid Fourier transform defined in Ref.~\onlinecite[Eq. 3.13]{BB}, and obtain, for $\underline{R}$ in the dual lattice $\R^*$ and $\eta\in \RR$,
$$2i\pi\left( \begin{array}{c}
	R_1\\
	R_2\\
	\eta
\end{array}\right)\cdot\widehat{W}(R,\eta)=(\widehat{\rho_0-m_0})(R,\eta).$$ 
For $\underline{R}=0$ and $\eta=0$, we obtain $\displaystyle\int_\Gamma \rho_0-m_0=0$. It follows that $\rho_0-m_0\in Y_\per$. 
Thus by Proposition~\ref{lemphi}, $\Phi_0:=G*_\Gamma (\rho_0-m_0)$ is \replaced{well-defined}{well defined}. Finally, we show that $\nabla \Phi_0=W$. Let  $\phi=(\phi_1,\phi_2,\phi_3)\in \D(\Gamma)^3 $. We have $\div \phi \in Y_\per$, thus $G*_\Gamma(\div\phi) \in L^\ii(\Gamma)$.
 Therefore 
\begin{align*}
    \langle \nabla G*_\Gamma(\rho_N-m_N),\phi \rangle &=  \mathbf{-} \langle \rho_N-m_N, G*_\Gamma\div\phi\rangle \\
    &\to \mathbf{-} \langle \rho_0-m_0, G*_\Gamma\div\phi\rangle = \langle \nabla G*_\Gamma(\rho_0-m_0),\phi \rangle.
\end{align*}
\end{proof}

\paragraph*{Step 3: Identification of the limit.}
We show \replaced{next}{in this section} that $\rho_0$ is invariant \replaced{w.r.t.}{with respect to} $(x_1,x_2)$ and it is indeed the unique minimizer of $\E_1^{m_0}$. This will prove that the convergences in Proposition~\ref{prop:cv-rhoN} hold for the whole sequence and prove points~{\it\ref{ii}-\ref{iii}} of Theorem~\ref{Th-H}. 

We start by showing that $\rho_0$ is the unique minimizer of $\E^{m_0}_\per$.

\begin{Pro}
$\rho_0$ is the unique minimizer of $\E^{m_0}_\per$. 
\label{prop:rho0-minimizer}
\end{Pro}

 \begin{proof}
Let $\rho\in \set{\rho\geq 0,\; \sqrt{\rho}\in X_\per,\; \int_\Gamma \rho=\int_\Gamma m_0 }$ \added{be} an admissible test function for $\E^{m_0}_\per$. Let us show that 
$$
\E_\per^{m_0}(\rho_0)\leq \E_\per^{m_0}(\rho). 
$$
As $\int_\Gamma m_0=\int_\Gamma m_N $ \added{by Lemma~\ref{Neutrality}}, $\rho$ is also an admissible test function for $\E_\per ^{m_N}$. Thus 
$$
I_N=\E_\per^{m_N}(\rho_N)\leq \E_\per^{m_N}(\rho). 
$$
We know from Proposition~\ref{prop:cv-rhoN} that 
 	$$
 	\E^{m_0}_\per(\rho_0) \leq \liminf I_N\leq \liminf\E_\per^{m_N}(\rho). 
 	$$
It thus remains to show that $\E_\per^{m_N}(\rho)\to \E_\per^{m_0}(\rho)$, which boils  down to showing that 
\begin{equation}\label{eq:cv-D-rho-mN}
D_G(\rho-m_N, \rho-m_N) \to D_G(\rho-m_0, \rho-m_0). 
\end{equation} 
We recall that 
$$
D_G(f,g)=\int_{\Gamma\times\Gamma} f(x)f(y)G(x-y)\diff x \diff y,
$$
with $G(x)=-2\pi \av{x_3}+\frac{1}{\av{x}}+\psi(x)$ and $\psi\in L^\ii(\Gamma)$. 
We denote by $h_N=\rho-m_N$ and $h=\rho-m_0$. 
We recall as well that for any $x_3\in \RR$, $h_N(\cdot,x_3)$ converges to $h(\cdot,x_3)$ weakly in $L^p(Q)$, \replaced{for all}{$\forall$} $ 1\leq p< \ii$. Therefore, for any $x_3,y_3\in\RR$,  $h_N(\cdot,x_3)h_N(\cdot,y_3)$ converges to  $h(\cdot,x_3)h(\cdot,y_3)$ weakly in $L^p(Q\times Q)$.
We have $-2\pi \av{x_3-y_3}+\psi(x-y) \in L^\ii(Q\times Q)$, and $\frac{1}{\av{x-y}}\in L^p(Q\times Q)$ for any $1\leq  p < 2$. Therefore, for any $x_3,y_3\in \RR$
\begin{equation}\label{eq:cv-ae}
\int_{Q\times Q} h_N(x) h_N(y)G(x-y) \diff\underline{x} \diff\underline{y}\to  \int_{Q\times Q} h(x) h(y)G(x-y)\diff\underline{x} \diff\underline{y}. 
\end{equation}
To use a dominated convergence argument, we need to bound the LHS of~\eqref{eq:cv-ae} uniformly \replaced{w.r.t.}{with respect to} $N$ by an $L^1$ function \replaced{w.r.t.}{with respect to} $(x_3,y_3)$. We start with the term with $\psi$. We have 
\begin{equation*}
\av{\int_{Q\times Q}  h_N(x) h_N(y)\psi(x-y) \diff\underline{x} \diff\underline{y} }\leq \norm{\psi}_{L^\ii(\Gamma)}\int_Q\av{h_N(x)}\diff\underline{x}\int_Q\av{h_N(y)}\diff\underline{y},
\end{equation*}
with
\begin{equation}\label{eq:borne-psi}
\int_Q \av{h_N}\diff \underline{x} \leq \int_Q {\rho}\diff \underline{x}+ \int_Q {m_N}\diff \underline{x}= \int_Q \rho\diff \underline{x}+ \int_Q {m}\diff \underline{x}. 	
\end{equation}
The RHS of~\eqref{eq:borne-psi} is indeed an $L^1$ function \replaced{w.r.t.}{wrt} to $x_3$ and $y_3$. We move to the term with $\av{x_3}$. We have
\begin{align}\label{eq:borne-x3}
\av{\int_{Q\times Q}  h_N(x) h_N(y)\av{x_3-y_3} \diff\underline{x} \diff\underline{y} }&\leq \int_Q\av{x_3 \,h_N(x)}\diff\underline{x}\int_{Q} \av{h_N(y)}\diff \underline{y}+\int_Q\av{y_3 \,h_N(y)}\diff\underline{y}\int_{Q} \av{h_N(x)}\diff \underline{x}\nonumber\\
&\leq  \bra{\int_Q \av{x_3}\bra{\rho(x)+m(x)}\diff \underline{x} }\bra{\int_Q\bra{\rho(y)+m(y)}\diff \underline{y} }\nonumber \\
& \quad \quad +\bra{\int_Q \av{y_3}\bra{\rho(y)+m(y)}\diff \underline{y} }\bra{\int_Q\bra{\rho(x)+m(x)}\diff \underline{x} }.
\end{align}
Again, the RHS of~\eqref{eq:borne-x3} is an $L^1$ function \replaced{w.r.t.}{wrt} $x_3$ and $y_3$. Finally, for the term with $\frac{1}{\av{x}}$, we split the term as follows
\begin{align}
\int_{Q\times Q} h_N(x)h_N(y)\frac{1}{\av{x-y}}\diff \underline{x} \diff \underline{y}&= \int_{Q\times Q} \rho(x)\bra{\rho(y)-2m_N(y)}\frac{1}{\av{x-y}}\diff \underline{x} \diff \underline{y} \nonumber \\
& \quad \quad \quad \quad \quad \quad+ \int_{Q\times Q} m_N(x)m_N(y)\frac{1}{\av{x-y}}\diff \underline{x} \diff \underline{y}. \nonumber
\end{align}
The first term is bounded by 
\begin{align*}
\av{ \int_{Q\times Q} \rho(x)\bra{\rho(y)-2m_N(y)}\frac{1}{\av{x-y}}\diff \underline{x} \diff \underline{y} }&\leq \norm{\rho*\frac{1}{\av{x}}}_{L^\ii(\Gamma)}\norm{\rho+2m(\cdot, x_3)}_{L^1(Q)},
\end{align*}
which is an $L^1$ function \replaced{w.r.t.}{with respect to} $x_3$. For the second term we use Young inequality  with $p= q=r=3/2$. We obtain
\begin{align}\label{eq:borne-1-x}
\av{\int_{Q\times Q}m_N(x)m_N(y)\frac{1}{\av{x-y}}\diff \underline{x}\diff \underline{y}}&\leq \norm{m_N(\cdot, x_3)}_{L^{3/2}(Q)}\norm{m_N(\cdot,y_3)}_{L^{3/2}(Q)}\norm{\frac{1}{\av{x}}}_{L^{3/2}(Q)}\nonumber\\
&\leq \norm{m(\cdot,x_3)}_{L^{3/2}(Q)}\norm{m(\cdot,y_3)}_{L^{3/2}(Q)} \times  \norm{\frac{1}{\av{\underline{x}}}}_{L^{3/2}(Q)}. 
\end{align}
As $m$ is compactly supported in $x_3$, $x_3\mapsto \norm{m(\cdot,x_3)}_{L^{3/2}(Q)}$ is also compactly supported and it is an $L^{3/2}$ function, thus it is \deleted{an} $L^1$; which concludes the proof.  
\end{proof}

\begin{Pro}
    The density $\rho_0$ is invariant \replaced{w.r.t.}{wrt} $(x_1,x_2)$, it is indeed the unique minimizer of the 1D model $\E_1^{m_0}$ and 
    $$
    I_0=\E_\per^{m_0}(\rho_0)=\E_1^{m_0}(\rho_0). 
    $$
    \label{1D}
\end{Pro}

\begin{proof}
 For any $\underline{R}\in \RR^2$, $\tau_{\underline{R}}\rho_0$ is also a minimizer of $\E_\per^{m_0}$, where $\tau$ is the translation operator. Thus, by the convexity of the functional  $\E_\per^{m_0}$ and the uniqueness of its minimizer (see Theorem~\ref{Th2}), we deduce that
 $\rho_0$ is invariant \replaced{w.r.t.}{with respect to} $(x_1,x_2)$. 
 To conclude the proof of the proposition, we notice that for any $\RR^2$-invariant function $f$
 $$
 D_G(f,f)=D_1(f,f),
 $$
 as shown in Ref.~\onlinecite[Prop. 3.1]{hb}. 
\end{proof}

\paragraph*{Step 4: Convergence of the energy.}
We prove in this section point~{\it \ref{i}} of Theorem~\ref{Th-H}, which consists of the convergence of the energy. 

\begin{Pro}
We have
$$ I_0=  \lim I_N. $$
\end{Pro}

\begin{proof}
By Propositions~\ref{prop:cv-rhoN} and~\ref{1D}, we have 
$$
I_0\leq \liminf I_N. 
$$
To prove the lower bound on $I_0$, we use $\rho_0$ as a test function for $\E_\per^{m_N}$ and obtain 
\begin{equation*}
    \begin{split}
    I_N\leq \int_{\Gamma } \abs{\nabla \sqrt{\rho_0}}^2 +\int_{\Gamma }\rho_0^{5/3}+D_G(\rho_0-m_N,\rho_0-m_N)\to \E_\per^{m_0}(\rho_0)=I_0,
\end{split}
\end{equation*}
where the last convergence is obtained using~\eqref{eq:cv-D-rho-mN}. Thus 
$$
\limsup I_N \leq I_0,
$$
which concludes the proof of the proposition. 

\end{proof}

\section{Numerical illustration}\label{sec:numerics}

In this section, we illustrate the convergence result of Theorem~\ref{Th-H}. We recall that the minimization problem  \eqref{TFW-per} has a unique minimizer \replaced{$\rho$ and that $u:=\sqrt{\rho}$}{$\rho$, and $u:=\sqrt{\rho}$} is the unique solution of the corresponding Euler--Lagrange equations 
\begin{equation}\label{base2}
\left\{	
\begin{array}{rll}
 \Delta u+\frac53 u^{\frac73}+ u\Phi & =& -\lambda u, \\
-\Delta \Phi & = & 4\pi(u^2-\mu).
\end{array} \right.
\end{equation}
We numerically solve the 3D periodic Euler--Lagrange system in \eqref{base2} by coupling a spectral approach and a fixed point iterative scheme
(a python implementation can be accessed at \protect\url{https://github.com/sa3dben/TFW_model}). We adopt the following iterative fixed point scheme: starting form an arbitrary initial guess $u_0$, at each fixed point iteration, a linear eigenvalue problem is solved  
\begin{equation}
\label{base2_lin}
\left( \Delta + \frac{5}{3} u_ n ^{4/3}+ \Phi_n\right)u_{n+1}=-\lambda_1 u_{n+1},
\end{equation}
with $\Phi_n$ the solution to the Poisson equation $-\Delta \Phi_n  =  4\pi(u_n^2-\mu)$.
The fixed point iterations are continued until convergence, i.e. $\left\|u_n-u_{n+1} \right\|_2 \leq \varepsilon$ for some tolerance $\varepsilon>0$. We take $\varepsilon = 10^{-6}$ in our simulations.

We adopt a periodic setting, with the unit cell $\textstyle \Gamma=Q\times [-\frac{L}2,\frac{L}2]$, and $Q = \left[-\frac{1}{2},\frac{1}{2}\right]^2$. We use a Fourier series decomposition approach to solve the linearized equation \eqref{base2_lin}.
The solution is expected to be $\R-$periodic \replaced{w.r.t.}{with respect to} $x_1$ and $x_2$ and decaying at infinity \replaced{w.r.t.}{with respect to} $|x_3|$. Taking a large value for $L$, it is reasonable to impose period conditions in the third direction as well. Therefore, we consider the Fourier basis
$$
f_k(x)= \frac1{\sqrt{L}} e^{i 2\pi \bra{\underline{x}\cdot\underline{k}+\frac{k_3 x_3 }{L}}}.
$$
We illustrate the convergence results in Theorem \ref{Th-H} using the series of nuclei densities $m_N$ defined as
$$m_N(x_1,x_2,x_3) = \mu(N,x_1) m_0(x_3) = \frac{5\pi}{2} \left|\cos(N  \pi x_1 )\right|  \exp\left(-\frac{x_3^2}{8}\right).$$

\noindent We note that $\int_Q \mu(N, x_1) \diff x_1 \diff x_2 = \frac{\pi}{2} \int_Q \left|\cos(N \cdot \pi x_1 )\right| = 1$. We take $L=2\pi$, which  proves effective as the function $m_0(x_3) = 5 \exp\left(-\frac{x_3^2}{8}\right)$ rapidly decreases in the $x_3$ direction. The same behaviour is assumed for the solutions $u_N$, and is validated with 1D simulations for $(u_0, \rho_0)$, the 1D solution corresponding to $m_0$.

The different functions are sampled on a grid of dimensions $(200 \times N, 4, 300)$. We keep the discretization constant along the \replaced{$x_2$ and $x_3$-axis}{$y$ and $z$-axis} as the homogenization is only applied  in the $x_3$ dimension. As we will only compute the zero Fourier mode along $x_2$, four points are enough. 
We compute the positive Fourier modes up to indices $K = (K_1, K_2, K_3) = (4 \times N, 0, 6)$. On the $x_3$-axis, we use the first $7$ Fourier modes, found to be sufficient for accurate reconstruction of $m_N$ and $m_0$ along that axis. Simulations in 1D prove that the solution $(u_0(x_3), \rho_0(x_3))$ is also accurately represented by these $7$ modes. We assume that this is also the case in 3D for all values of $N$.  The choice of $(4N+1)$ Fourier modes in the $x_1$ direction, is motivated by the need for increased Fourier modes as $N$ grows, to capture the changing characteristics and new harmonics for $m_N$ and eventually for the solution $(\rho_N, u_N)$. In fact, we know that $(\rho_N, u_N)$ will rather converge to a constant along the axis $x_1$ by Theorem \ref{Th-H}, however we enforce a hard check for this by computing also large modes for $\rho_N$ and $u_N$ .

Different $L^p$-norms of the error $e_N = \rho_N - \rho_{0}$ are presented in Figure \ref{fig: e_n norms} for values of $N=1,\ldots,5$. We, indeed, observe the convergence stated in Theorem~\ref{Th-H}.
The same figure shows some convergence rates estimates. We conjecture that we have a theoretical convergence rate of $1/N^r$ with $\frac{10}{3}<r< 4$ for the different norms of \replaced{$e_N$}{$e_n$} and the gradient  $L^2$ norm $\norm{ \nabla u_{\text{\replaced{$N$}{$n$}}} -\nabla{u_{0}} }_2$. The discrepancies in Figure \ref{fig: e_n norms} at $N=4,5$ are probably due to numerical errors.

\begin{figure}[ht!]
    \centering
    \includegraphics[scale=0.3]{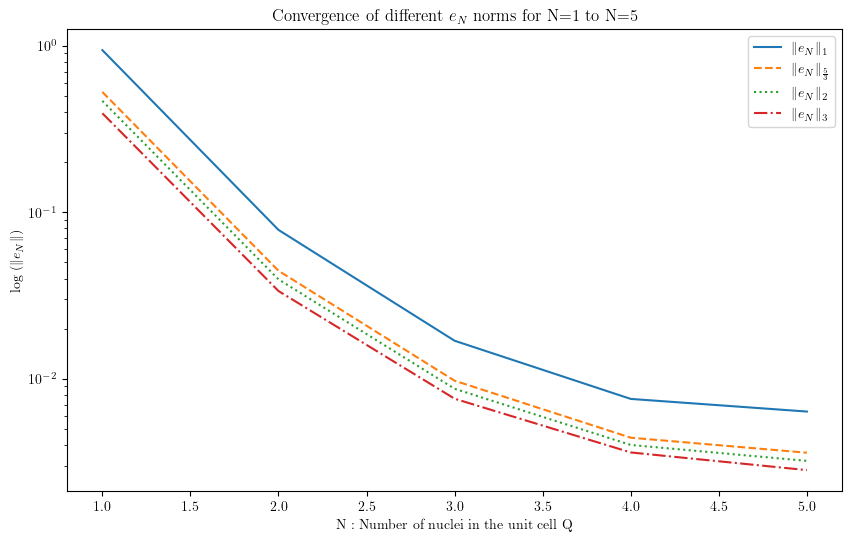}
    \includegraphics[scale=0.3]{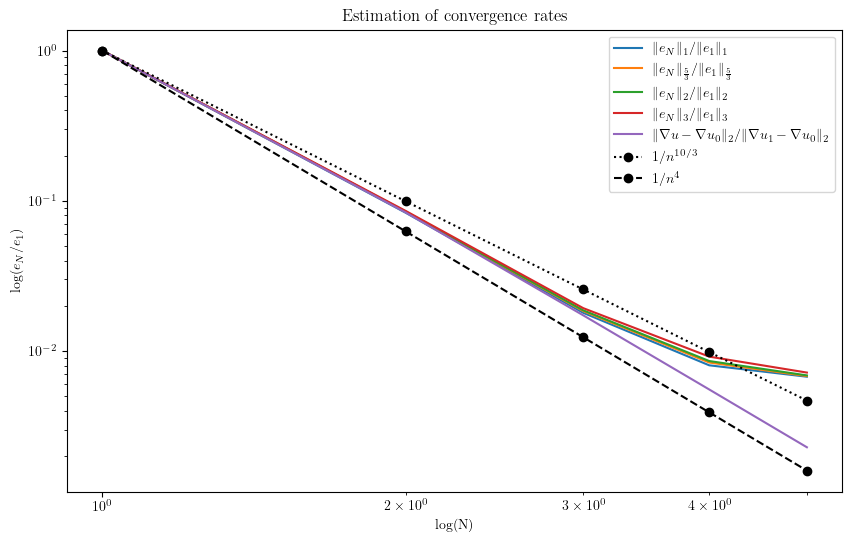}
    \caption{ Convergence analysis for \replaced{$e_N$}{$e_n$} norms (left) and  convergence rates estimation for \replaced{$e_N$}{$e_n$} norms and the gradient $L^2$ norm $\norm{ \nabla{u_n} -\nabla{u_{0}} }_2$ 
 (right).}
    \label{fig: e_n norms}
\end{figure}

Figure \ref{fig: energy norms} illustrates the convergence of the ground state energy  $I_N =\E^{m_N}_{per}(\rho_N)$ to the 1D energy $I_0 =\E^{m_0}_{1}(\rho_0) $, proved in Theorem~\ref{Th-H}. The estimated convergence rate for the energy is of the order~$\approx \dfrac{1}{N^{3/2}}$.

\begin{figure}[ht!]
    \centering
    \includegraphics[scale=0.35]{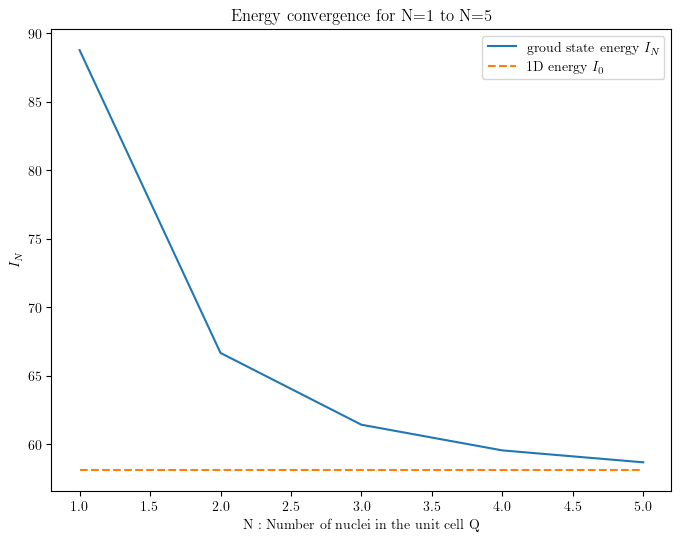}
    \includegraphics[scale=0.35]{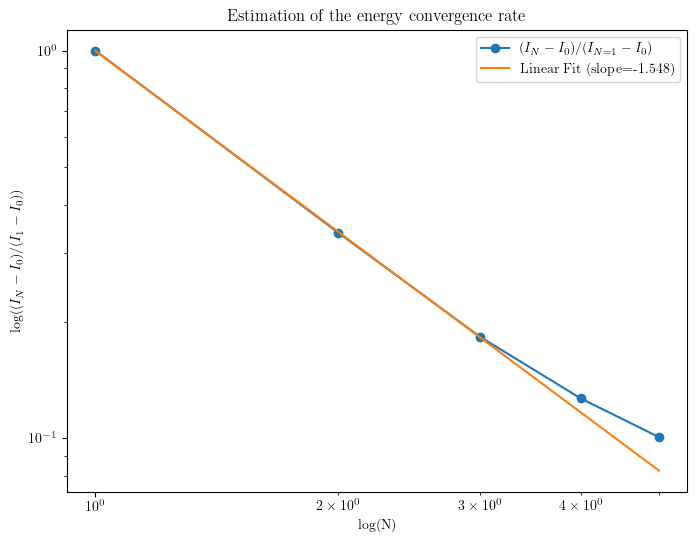}
    \caption{Convergence analysis for the energies $I_N$ (left) and  convergence rate estimation (right).}
    \label{fig: energy norms}
\end{figure}

\begin{Rem}
It is noteworthy to mention that, for $N \geq 6$, and using higher order Fourier modes, the presence of numerical errors hinders a clear observation of further convergence for $e_N$.
\end{Rem}

\bibliographystyle{siam} 
\bibliography{main}

\end{document}